\documentclass[12pt,a4paper]{amsart}

\usepackage{amsmath,amsfonts,amssymb}

\usepackage[hmargin=2cm,vmargin=2cm]{geometry}

\usepackage[hyperfootnotes=false,colorlinks=true,linkcolor=blue,%
citecolor=purple,filecolor=magenta,urlcolor=cyan]{hyperref}
\usepackage{nameref,zref-xr}                    

\usepackage{amsthm,amscd}
\usepackage{epsfig}
\usepackage{color}
\usepackage[all]{xy}
\usepackage{tikz}
\usepackage{graphics, setspace}
\usepackage{breqn}
\usepackage{amstext}
\usepackage{array}

\newcommand{\CC}{{\mathbb{C}}}

\newcommand{\PP}{{\mathbb{P}}}
\newcommand{\QQ}{{\mathbb{Q}}}

\def\A{{\mathcal A}}

\def\H{{\mathcal H}}

\def\O{{\mathcal O}}

\def\T{{\mathcal T}}

\def\res{{\mathrm{res}}}

\newcommand{\p}{{\partial}}

\newcommand{\eps}{\varepsilon}

\newcommand{\bt}{{\bf t}}

\newtheorem{theorem}{Theorem}[section]
\newtheorem*{theorem*}{Theorem}
\newtheorem{proposition}[theorem]{Proposition}
\newtheorem{lemma}[theorem]{Lemma}
\newtheorem{corollary}[theorem]{Corollary}
\newtheorem*{corollary*}{Corollary}

\newtheorem{example}[theorem]{Example}
\newtheorem{remark}[theorem]{Remark}

\usepackage{color}

\def\&{\vspace{-5pt}&}

\makeatletter
\newsavebox{\@brx}
\newcommand{\llangle}[1][]{\savebox{\@brx}{\(\m@th{#1\langle}\)}%
  \mathopen{\copy\@brx\kern-0.5\wd\@brx\usebox{\@brx}}}
\newcommand{\rrangle}[1][]{\savebox{\@brx}{\(\m@th{#1\rangle}\)}%
  \mathclose{\copy\@brx\kern-0.5\wd\@brx\usebox{\@brx}}}
\makeatother

\allowdisplaybreaks

\numberwithin{equation}{section}

\begin{document}

\title{
Genus zero Whitham hierarchy via Hurwitz--Frobenius manifolds
}

\author{Alexey Basalaev}
\address{A. Basalaev:
\newline 
Faculty of Mathematics, HSE University, Usacheva str., 6, 119048 Moscow, Russian Federation
}
\email{abasalaev@hse.ru}

\date{\today}

 \begin{abstract}
 B.~Dubrovin introduced the structure of a Dubrovin--Frobenius manifold on a space of ramified coverings of a sphere by a genus $g$ Riemann surface with the prescribed ramification profile. This is now known as a genus $g$ Hurwitz--Frobenius manifold.
 We investigate the genus zero Hurwitz--Frobenius manifolds and their connection to the integrable hierarchies. In particular, we prove that the Frobenius potentials of the genus zero Hurwitz--Frobenius manifolds stabilize and therefore define an infinite system of commuting PDEs.

 We show that this system of PDEs is equivalent to the genus zero Whitham hierarchy of I.~Krichever.
 Our result shows that this system of PDEs has both Fay form, depending heavily on the flat structure fo the Hurwitz--Frobenius manifold and coordinatefree Lax form. We also show how to extend this system of PDEs to the multicomponent KP hierarchy via the $\hbar$--deformation of the differential operators.
 \end{abstract}
 \maketitle

 
\section{Introduction}
Tense connection between the Dubrovin--Frobenius manifolds and the integrable hierarchies had been known since the early 90s. This includes the general construction of Dubrovin--Zhang \cite{DZ} and Buryak \cite{Bu} and the more specific examples including \cite{DVV,D1}, \cite{FGM, MT}, \cite{LRZ, FSZ, DLZ}, \cite{BDbN,B22,B24}.
We will be interested in the specific examples because they assume the Dubrovin--Frobenius manifold to be in some way special.
The references given are groupped by the approach used to connect the Dubrovin--Frobenius manifolds with the integrable hierarchies. The first one uses Lax operators, the second Hirota-like bilinear equations, the third bihamiltonian structure and the last one the Fay--type equations. These approaches are very different and not totally equal. In particular, the Lax operator approach was only applied for A--type Dubrovin--Frobenius manifolds. The bilinear equations form is only know for few examples comming from Saito theory. The Fay--type approach above can be only applied to a family of Dubrovin--Frobenius manifolds that satisfies the certain stabilization condition.

In this note we extend both the Lax operator approach of \cite{DVV,D1} and the Fay--type approach of \cite{BDbN}. In particular, we consider the specific Dubrovin--Frobenius manifolds called genus zero Hurwitz--Frobenius manifolds and show that both approaches can be used for them giving the Whitham hieararchy introduced by I.~Krichever \cite{K94}.

\subsection*{Dubrovin--Frobenius manifolds} were introduced by B.~Dubrovin in the early 90s (cf. \cite{D2}). This is a compex manifold $M$ equipped with an associative and commutative product $\circ: \T_M \otimes_M \T_M \to \T_M$ on the holomorphic tangent sheaf, a flat pairing $\eta: \T_M \otimes_M \T_M \to \O_M$ and a flat unit vector field. These data all together should satisfy the certain integrability condition. In particular, there should exist the special coordinates $t^1,\dots,t^n$, called \textit{flat} and a special function $F \in \O_M$, called \textit{potential}, such that the unit vector field coincides with $\frac{\p}{\p t^1}$ and
\[
    \eta( \frac{\p}{\p t^\alpha},\frac{\p}{\p t^\beta}) = \frac{\p^3 F}{\p t^1 \p t^\alpha \p t^\beta }  = \eta_{\alpha,\beta} \in \CC,
    \quad
    \frac{\p}{\p t^\alpha} \circ \frac{\p}{\p t^\beta} = \sum_{\gamma,\delta=1}^n \frac{\p^3 F}{\p t^\alpha \p t^\beta \p t^\gamma} \eta^{\gamma,\delta} \frac{\p}{\p t^\delta}.
\]

The potential is subject to huge system of PDEs called WDVV equation. Sometimes the Dubrovin--Frobenius manifold itself is built from the potential given - like in Gromov--Witten theory. In this case the formulae above should be read from the right to the left.
In the ``geometric'' cases the pairing $\eta$ and the product $\circ$ are defined, but the potential and the flat coordinates have to be found.
Examples of such Dubrovin--Frobenius manifolds are Saito--Frobenius (cf. \cite{ST}) and Hurwitz--Frobenius manifolds.

\subsection*{Hurwitz--Frobenius manifolds} It was observed by B.~Dubrovin \cite[Lecture 5]{D2} that the space $\H_{g,K}$ of ramified coverings of $\PP^1$ by a genus $g$ Riemann surface with the prescribed ramification profile $K$ can be endowed with a structure of Dubrovin--Frobenius manifold. Such structures were investigated in \cite{S,M17,B14,PS} etc..
B. Dubrovin constructs $\eta$ and $\circ$ via the certain residue calculus.

In this text we consider the space $\H_{0,K}$ with $g=0$ and $K = \lbrace k_0,\dots,k_N \rbrace$ of $\lambda: \PP^1 \to \PP^1$
\[
 \lambda(z) := z^{k_0} + \sum _{a=1}^{k_0-1} v_{0,a} z^{a-1} + \sum _{i=1}^N \sum _{a=1}^{k_i} v_{i,a-1} \left(z-q_i\right)^{-a}
\]
where $v_{i,a}$ and $q_i$ are the coordinates of $\H_{0,K}$. These coordinates look essential, however these are not flat. The flat coordinates $t_{i,\alpha} = t_{i,\alpha}(v_{i,\bullet})$ are obtained via the certain residue calculus involving the function $\lambda$.

The Hurwitz--Frobenius manifold $\H_{0,k_0}$ (with $K$ of length $1$ and $N = 0$) is isomorphic to the Dubrovin--Frobenius manifold of type $A_{k_0-1}$. The function $\lambda$ coincides then with the symbol of Lax operator of $(k_0-1)$--KdV what was used extensively in \cite{DVV}.

The first theorem of this paper is the following.
Consider $\H_{0,K}$ and $\H_{0,L}$ with $K = \lbrace k_0,k_1,\dots,k_N \rbrace$ and $L = \lbrace l_0,l_1,\dots,l_N \rbrace$. Note that the lengths of the sets $K$ and $L$ are assumed to be the same. We will say that $L \ge K$ if $l_j \ge k_j$ for every index $j$.

\begin{theorem*}[Theorem~\ref{theorem: stabilization} below]
    Let $L \ge K$, denote by $F^L$ and $F^K$ the potentials of $\H_{0,L}$ and $\H_{0,K}$ respectively. For any $0 \le i,j \le N$ let $\alpha \le k_i$, $\beta \le k_j$ if $i \neq j$ and $\alpha+\beta \le k_i-2$ if $i=j$.

    Then
    \[
        \frac{\p^2 F^{K}}{\p t_{i,\alpha} \p t_{j,\beta}} \mid_{t_{r,\gamma} \ = \ - k_r s_{r,k_r-\gamma}}
        =
        \frac{\p^2 F^{L}}{\p t_{i,\alpha} \p t_{j,\beta}} \mid_{t_{r,\gamma} \ = \ - k_r s_{r,k_r-\gamma}}
    \]
    assumed as the equality of function in $s_{r,\bullet}$ and $q_1,\dots,q_N$.
\end{theorem*}
In \cite{BDbN} the series of Dubrovin--Frobenius manifolds satisfying the equality above was called \textit{stabilizing}. It was observed in loc.cit. that stabilizing Dubrovin--Frobenius manifolds can be used to introduce infinite systems of commuting PDEs.

Consider a function $f$ depending on $\lbrace t_{i,\alpha} \rbrace$ where $\alpha = 0,1,2,\dots$  if $i \ge 1$ and $\alpha = 1,2,\dots,$ if $\alpha = 0$. Denote $\p_{i,\alpha} := \frac{\p f}{\p t_{i,\alpha}} $ The corresponding system of PDEs on $f$ reads
\begin{equation}\label{eq: BDbN form}
     \p_{i,\alpha} \p_{j,\beta} f = \frac{\p^2 F^K}{\p t_{i,\alpha} \p t_{j,\beta}} \mid_{t_{r,\gamma} \ = \ -k_r \p_{0,1} \p_{r,k_r-\gamma} f },
\end{equation}
where the set $K$ should be taken such that $\alpha \le k_i$, $\beta \le k_j$ and $\alpha+\beta \le k_i-2$ if $i=j$. The theorem above shows that any set $K$ satisfying this property will result in the same expression on the right hand side. This follows from WDVV equation that this systemd of PDEs commutes.

\begin{corollary*}
    The series of Hurwitz--Frobenius manifolds $\H_{0,K}$ defines a system of infinite commuting PDEs.
\end{corollary*}

\subsection*{Whitham hierarchy}

I.~Krichever introduced in \cite{K94} the dispersionless hierarchy associated to the space of ramified coverings as above. This hierarchy was investigated in many papers including \cite{D2,GMMM,TT07,Za} etc.. The definition of I.~Krichever was given in the bilinear equation form however is also has the following Lax form, generalizing the idea of \cite{DVV}.

Let $z_0,z_1,\dots,z_N$ be infinite Laurent series at the punctures
\[
  z_0 = z + \sum_{j = 2}^\infty u_{0j} z^{-j+1} , \quad z_i = \frac{r_i}{z - q_i} + \sum_{j = 1}^\infty u_{i j} (z-q_i)^{j-1}
\]
for some $r_i$ and $u_{ij}$ depending on $v_{ij}$. Time evolution of the hierarchy is generated by
\[
    \p_{\alpha n} z_i = \lbrace \Omega_{\alpha n} (z), z_i \rbrace
\]
with the Poisson bracket
\[
 \lbrace f,g \rbrace = \frac{\p f}{\p z} \frac{\p g}{\p t_{01}} - \frac{\p f}{\p t_{01}} \frac{\p g}{\p z}
\]
and the Hamiltonians $\Omega_{\alpha n}$ given by the equalities
\[
    z_0^n = \Omega_{0n}(z) + O(z^{-1}), \ z \to \infty,
    \quad
    z_\alpha^n = \Omega_{\alpha n}(z) + O(1), \ z \to q_\alpha.
\]
This hierarchy is dispersionless by its definition.

Takasaki and Takebe gave in \cite{TT07} the \textit{Fay form} of genus zero Whitham hierarchy.
For $i =0,1,\dots$ consider the operators
\[
 D_i(z) := \sum_{\alpha = 1}^\infty \frac{z^{-\alpha}}{\alpha} \p_{i,\alpha}, \ \text{ with } \ \p_{i,\alpha} := \frac{\p}{\p t_{i,\alpha}}.
\]
Let the function $f$ depend on $\lbrace t_{i,\alpha} \rbrace$ where $\alpha = 0,1,2,\dots$  if $i \ge 1$ and $\alpha = 1,2,\dots$ if $\alpha = 0$.
The following system of equations in the ring of formal power series in $z^{-1}$ and $w^{-1}$ is the way to write an infinite system of PDEs on the function $f$.
\begin{align*}
    e^{D_0(z)D_0(w)f} &= 1 - \frac{\p_{0,1} (D_0(z) - D_0(w)) f}{z - w},
    \\
    z e^{D_0(z)(\p_{i,0} + D_i(w))f} &= z - \p_{0,1} (D_0(z) - \p_{i,0} - D_i(w)) f,
    \\
    e^{(\p_{i,0} +D_i(z))(\p_{i,0} +D_i(w))f} &= - \frac{zw \p_{0,1} (D_i(z) - D_i(w)) f}{z - w},
    \\
    e^{(\p_{i,0} + D_i(w))(\p_{j,0} + D_j(w))f} &= - \p_{0,1} (\p_{i,0}+ D_i(z) - \p_{j,0} - D_j(w)) f, \quad i \le j.
\end{align*}
Takasaki and Takebe proved it to be equivalent to genus zero Whitham hierarchy.

\begin{theorem*}[Theorem~\ref{theorem: Fay} below]
    Consider the Hurwitz--Frobenius manifold $\H_{0,k_0,k_1,\dots,k_N}$ with the potential $F$ and denote $m := \min(k_0,k_1,\dots,k_N)$.

    After the rescaling $t_{i,\alpha} \mapsto \alpha t_{i,\alpha}$ with $ i \ge 0$ and $\alpha \ge 1$ the potential $F$ satisfies the Fay form of the genus zero Whitham hierarchy up to total order $-m$ of the formal variables $z$ and $w$.
\end{theorem*}

\begin{corollary*}
    The infinite system of PDEs defined by the stabilizing series of Hurwitz--Frobenius manifolds $\H_{0,K}$ coincides with the genus zero Whitham hierarchy.
\end{corollary*}

It's important to note that our theorem connects two ways to introduce an integrable system associated to $\H_{0,K}$: the Lax form - as given by I.~Krichever and the Fay form - as obtained from the potentials of $\H_{0,K}$.

\begin{remark}
Genus zero Whitham hierarchy was investigated in the context of Dubrovin--Frobenius manifolds in \cite{Sh24,SWZ}, however in a totally different approach by assuming the new concept called ``infinite--dimensional Frobenius manifold''. 
\end{remark}

\subsection*{Dispersionfull hierarchy}
It was also found in \cite{TT07} that the genus zero Whitham hierarchy in the dispersionless limit of the multicomponent KP hierarchy (see Theorem~2 of loc.cit.). Their result also helps us to construct the dispersionfull hierarchy from Eq.~\eqref{eq: BDbN form}. This is obtained via the so--called $\hbar$--deformation of the differential operators $D_i$. We explain it in Section~\ref{section: h-deformation}.

\subsection{Acknowledgements}
The work of Alexey Basalaev was supported by the Russian Science Foundation (grant No. 24-11-00366).

\section{Hurwitz--Frobenius manifold}

\subsection{Dubrovin--Frobenius manifolds}
Assume $M$ to be an open full--dimensional subspace of $\CC^l$. We say that it's endowed with a structure of Dubrovin--Frobenius manifold if there is a regular function $F = F(t_1,\dots,t_l)$ on $M$, s.t. the following conditions hold (cf. \cite{D2}).

\begin{itemize}
	\item 
There is a distinguished variable $t_k$, for some $1 \le k \le l$, such that:
\[
    \frac{\p F}{\p t_k} = \frac{1}{2} \sum_{\alpha,\beta = 1}^l \eta_{\alpha,\beta} t_\alpha t_\beta,
\]
and $\eta_{\alpha,\beta}$ are components of a non-degenerate bilinear form $\eta$ (which does not depend on $t_\bullet$). In what follows denote by $\eta^{\alpha,\beta}$ the components of $\eta^{-1}$.

\item The function $F$ satisfies a large system of PDEs called the WDVV equations:
\[
\sum_{\mu,\nu = 1}^l \frac{\p^3 F}{\p t_\alpha \p t_\beta \p t_\mu} \eta^{\mu,\nu} \frac{\p^3 F}{\p t_\nu \p t_\gamma \p t_\sigma}
=
\sum_{\mu,\nu = 1}^l \frac{\p^3 F}{\p t_\alpha \p t_\gamma \p t_\mu} \eta^{\mu,\nu} \frac{\p^3 F}{\p t_\nu \p t_\beta \p t_\sigma},
\]
which should hold for every given $1\leq\alpha,\beta,\gamma,\sigma\leq l$.

\item There is a vector field $E$ called the \textit{Euler vector field}, such that modulo quadratic terms in $t_\bullet$ we have $E \cdot F = (3-\delta) F$ for some fixed complex number $\delta$. We will assume $E$ to have the following simple form
\[
    E = \sum_{i=1}^l d_i t_i \frac{\p}{\p t_i}
\]
for some fixed numbers $d_1,\dots,d_l$. 
\end{itemize}

Given such a data $(M,F,E)$ one can endow every tangent space $T_pM$ with a structure of commutative associative product $\circ$ (depending on $\bt$) defined as follows:
\[
    \frac{\p}{\p t_\alpha} \circ \frac{\p}{\p t_\beta} = \sum_{\delta,\gamma=1}^l\frac{\p^3 F}{\p t_\alpha \p t_\beta \p t_\delta} \eta^{\delta\gamma} \frac{\p}{\p t_\delta}.
\]
The unit of this product is the vector field $e = \frac{\p}{\p t_k}$.
It follows that $\eta(a \circ b,c) = \eta(a,b \circ c)$ for any vector fields $a,b,c$.

\subsection{Hurwitz--Frobenius manifold}
Consider the space of the meromorphic functions $\lambda(z)$ on $\CC$ of the form
\begin{equation*}
 \lambda := z^{k_0} + \sum _{a=1}^{k_0-1} v_{0,a} z^{a-1} + \sum _{i=1}^N \sum _{a=1}^{k_i} v_{i,a-1} \left(z-q_i\right)^{-a},
\end{equation*}
where $v_\bullet$ and $q_\bullet$ are complex numbers, such that
\[
    v_{1,k_1-1},\dots, v_{N,a_{k_N}-1} \neq 0, \quad \text{and} \quad q_{i},q_{j} \ \text{pairwise distinct}.
\]
Let $M$ stand for the space of all such $\lambda$, parametrized by $v_\bullet$ and $q_\bullet$. This is a compelex manifold of dimension
\[
    k_0-1 + \sum_{i=1}^N (k_i+1).
\]

Following B.~Dubrovin (see Lecture~5 of \cite{D2}) associate to $M$ the three-point function.
Denote by
\[
 \mathcal{R} := \lbrace \infty, q_{1},\dots, q_{N} \rbrace
\]
the set of ramification points of $\lambda$.
For $X \in \T_M$ let $X \cdot \lambda$ stand for the respective directional derivative. Set
\begin{equation}\label{eq: three-point function}
   \langle X,Y,Z \rangle := - \sum_{i=0}^N \res_{z \in \mathcal{R}} \Big( (X \cdot \lambda)(Y \cdot \lambda)(Z \cdot \lambda) \ \frac{dz }{\frac{\p \lambda}{\p z}} \ \Big), \quad \forall X,Y,Z \in \T_M,
\end{equation}

Define also the bilinear form $\eta: \T_M \otimes_M \T_M \to \O_M$ by
\[
    \eta(X,Y) := - \sum_{i=0}^N \res_{z \in \mathcal{R}} \Big( (X \cdot \lambda)(Y \cdot \lambda) \ \frac{dz }{\frac{\p \lambda}{\p z}} \ \Big), \quad \forall X,Y,Z \in \T_M.
\]
This bilinear form is symmetric by the definition but also non--degenerate, defining the pairing on $\T_M$.

\begin{remark}
    At this point we make the choice of the so--called ``primary differential'' of B.~Dubrovin. In our case this is just $dz$.
\end{remark}

The three--point function and the pairing above allow one to introduce the product $\circ: \T_M \otimes_M \T_M \to \T_M$ by
\[
    \eta(X \circ Y, Z) = \langle X,Y, Z \rangle, \quad \forall Z \in \T_M.
\]
This product introduces on $M$ a Dubrovin--Frobenius manifold structure. This follows immediately, that $e := \frac{\p}{\p v_{0,1}}$ is the unit of $\circ$ and $\eta(X,Y) = \langle e , X, Y \rangle$.

\begin{proposition}\label{prop: low ramification products}
  Let $X,Y,Z \in \T_M$ be such that $(X \cdot \lambda)(Y \cdot \lambda) = (Z \cdot \lambda)$
    assumed as the equality of rational functions. Then $X \circ Y = Z$.
\end{proposition}
\begin{proof}
    It follows from the definition of the three--point function that $\eta(Z,W) = \eta(X \circ Y, W)$ for any $W \in \T_M$. Because $\eta$ is non--degenerate, this holds if and only if $Z = X \circ Y$.
\end{proof}

The following simple observation plays an important role in what follows.
\begin{example}\label{example: v-products in same sector}
    We have
    \[
     \frac{\p}{\p v_{0,a}} \circ \frac{\p}{\p v_{0,b}} = \frac{\p}{\p v_{0,a+b-1}}, \quad \text{if} \ a + b \le k_0-1,
    \]
    and
    \[
     \frac{\p}{\p v_{i,a}} \circ \frac{\p}{\p v_{i,b}} = \frac{\p}{\p v_{0,a+b+1}}, \quad \text{if} \ a + b \le k_i-2.
    \]

\end{example}

The coordinates $v_\bullet$ look ``essential'' to parametrize the meromorphic functions with the prescribed ramification profile, however these are not flat for the pairing $\eta$ except for the simplest cases of low dimention. In order to define the potential we need to find the flat coordinates.

\subsection{Flat structure}\label{section: flat structure}

Let also $z_0,z_1,\dots,z_N$ be s.t. $\lambda = z_i^{k_i}$ in a neighborhood of $z = q_i$. We have
\[
  z_0 = z + \frac{v_{0,k_0-1}}{k_0} z^{-1} + \dots, \quad z_i = v_{i,k_i-1}^{1/k_i} (z - q_i)^{-1} + \dots.
\]

Then define the coordinates $t_\bullet$ by
\begin{align}
    &t_{0,\alpha} = \frac{k_0}{k_0-\alpha} \res_{z = \infty} z_0^{k_0-\alpha} dz,\
    \quad \alpha = 1,\dots,k_0-1,
    \label{eq: flat coordinates 0}
    \\
    &t_{i,0} = v_{i,0}, \ t_{i,k_i} = q_i \quad 
    t_{i,\alpha} = - \frac{k_i}{k_i-\alpha} \res_{z = k_i} z_i^{k_i-\alpha} dz,\
    \quad \alpha = 1,\dots,k_i, \ i \ge 1.
    \label{eq: flat coordinates i}
\end{align}

\begin{theorem}[Theorem 5.1 of \cite{D2}]
    In the cooridinates $t_\bullet$ above the three--point function and the pairing $\eta$ define the structure of a Dubrovin--Frobenius manifold on $\H_{0,k_0,k_1,\dots,k_N}$.
    
    The only non-zero pairings in the flat frame are
    \begin{align*}
        & \eta(\frac{\p}{\p t_{0,\alpha}},\frac{\p}{\p t_{0,k_0-\alpha}}) = -\frac{1}{k_0}, && 1 \le \alpha \le k_0-1,
        \\
        &\eta(\frac{\p}{\p t_{i,0}},\frac{\p}{\p t_{i,k_i}}) = -1, \ 
        \eta(\frac{\p}{\p t_{i,\alpha}},\frac{\p}{\p t_{i,k_i-\alpha}}) = -\frac{1}{k_i}, && i \ge 1, 1 \le \alpha \le k_i-1.
    \end{align*}
    
    The potential $F$ of this Dubrovin--Frobenius manifold is related to the three--point function (cf. Eq.~\eqref{eq: three-point function}) by
    \[
        \frac{\p^3 F}{\p t_{i,\alpha} \p t_{j,\beta} \p t_{r,\gamma}} = \left\langle \frac{\p}{\p t_{i,\alpha}},\frac{\p}{\p t_{j,\beta}},\frac{\p}{\p t_{r,\gamma}} \right \rangle.
    \]

     The Euler vector field reads
     \[
         E = \sum_{\alpha=1}^{k_0-1} 
         \left( \frac{k_0 + 1}{k_0} - \frac{\alpha}{k_0} \right)
         t_{0,\alpha} \frac{\p}{\p t_{0,\alpha}} +  \sum_{i=1}^N \sum_{\alpha=1}^{k_i}
         \left( \frac{k_0 + 1}{k_0} - \frac{\alpha}{k_i} \right) t_{i,\alpha} \frac{\p }{t_{i,\alpha}}.
     \]
     The potential $F$ is subject to the quasihomogeneity condition $E \cdot F = 2(1 + \frac{1}{k_0}) F$ modulo the quadratic terms.
         
\end{theorem}

    In this text we make an unorthodox variable choice so that the pairing in the flat basis gets negative signs. This is because our $t_{0,\bullet}$ variables are introduced with the reversed sign. Even though this looks unusual from the point of view of Dubrovin's legacy, this will suit well Whitham hierarchy applications (see Remark~\ref{remark: coordinates sign}).

    The theorem above fixes the potential $F$ only up to quadratic terms. We will fix the quadratic terms in the next section.

\subsection{Supplimentary facts}

\begin{lemma}\label{lemma: sd}
    Let $i,j \ge 0$ then
    \[
        \left(\frac{\alpha}{k_i} + \frac{\beta}{k_j} \right) \frac{\p^2 F}{\p t_{i,\alpha}\p t_{j,\beta}} = \sum_{r,\gamma} ( \frac{k_0+1}{k_0} - \frac{\gamma}{k_r} ) t_{r,\gamma} \langle \frac{\p}{t_{i,\alpha}}, \frac{\p}{t_{j,\beta}}, \frac{\p}{t_{r,\gamma}} \rangle.
    \]
%
\end{lemma}
\begin{proof}
    We get this immediately differentiating the quasihomogeneity condition on $F$ with respect to $t_{i,\alpha}$ and $t_{j,\beta}$.
\end{proof}

Lemma above does not allows one to express the second order derivatives of $F$ with respect to $t_{i,0}$ and $t_{j,0}$. These cases should be treated separately.

This follows immediately from the simple residue computations that for $i \neq j$
\[
    \frac{\p^3 F}{\p t_{i,0} \p t_{j,0} \p t_{r,\alpha}} = 
    \begin{cases}
        \dfrac{1}{t_{i,k_i}-t_{j,k_j}} \quad &\text{if } r=i,\ \alpha=k_i
        \\
        0 \quad & \text{otherwise},
    \end{cases}
\]
and
\[
    \frac{\p^3 F}{\p t_{i,0} \p t_{i,0} \p t_{r,\alpha}} = 
    \begin{cases}
        \dfrac{1}{t_{i,k_i-1}} \quad &\text{if } r=i,\ \alpha=k_i-1
        \\
        0 \quad & \text{otherwise}.
    \end{cases}
\]
In what follows we assume $F$ to satisfy
\begin{equation}\label{prop: simple 0 insertions}
     \p_{i,0}\p_{j,0} F = \log 
    (t_{i,k_i}-t_{j,k_j}), \quad \p_{i,0}\p_{i,0} F = \log  (t_{i,k_i-1}).
\end{equation}

\begin{proposition}
    In the flat coordinates we have
    \begin{align}
        & \frac{\p \lambda}{\p t_{0,\alpha}} = z \frac{\p \lambda}{\p t_{0,\alpha-1}} - \sum _{\beta=1}^{\alpha-2} \frac{t_{0,k_0-\beta}}{k_0}  \frac{\p \lambda}{\p t_{0,\alpha-1-\beta}}, \quad \alpha \le k_0-1
        \label{eq: lambda 1}
        \\
        & \frac{\p \lambda}{\p t_{i,\alpha}} = \frac{1}{\left(z-q_{i}\right)} \sum _{\beta=1}^{\alpha} \frac{t_{i,k_i-\beta}}{k_i} \frac{\p \lambda}{\p t_{i,\alpha-\beta}}, \quad i \ge 1, \ \alpha \le k_i-1.
        \label{eq: lambda 2}
    \end{align}
\end{proposition}
\begin{proof}
    If follows immediately from the residue computations that $t_{0,\alpha}$ depends on $v_{0,\alpha},\dots,v_{0,k_0-1}$ only. Then Eq.~\eqref{eq: lambda 1} is an equality of polynomials rathen than meromorphic functions on $\CC$. In this form it can be found in \cite[Section 4.3]{DVV}. 

    For any $i \ge 1$ and $\alpha < k_i$ any flat coordinate $t_{i,\alpha}$ depends on $v_{i,\alpha},\dots,v_{i,k_i-1}$ only. Then Eq.~\eqref{eq: lambda 2} is an equality of meromorphic functions on $\CC$ with the poles at $z = q_i$ only. This means that if we expand both sides of this equation in $z_i$, it is enough to prove it up to order $-k_i$.
    
    In a neighborhood of $z = q_i$ we have (cf. Lemma 5.2 of \cite{D2})
    \[
        \frac{\p \lambda}{\p t_{i,\alpha}} dz = - \frac{\p z}{\p t_{i,\alpha}} d\lambda = - (z_i^{\alpha} + O(1)) dz.
    \]
    Using Eq.~\eqref{eq: local variable expansion} again we have
    \[
        (z-q_i) \frac{\p \lambda}{\p t_{i,\alpha}} = - \frac{1}{k_i} \sum_{\beta=1}^{\alpha} t_{i,\beta} z_i^{-(k_i-\beta) + \alpha} + o(z_i^{-k_i}) = \frac{1}{k_i} \sum_{\beta=1}^{\alpha} t_{i,k_i-\beta} z_i^{\alpha-\beta}  + o(z_i^{-k_i})
    \]
    what completes the proof.
\end{proof}

\section{Stabilization}
In this section consider the Hurwitz--Frobenius manifolds $\H_{0,K}$ and $\H_{0,L}$ with $K = \lbrace k_0,k_1,\dots,k_N \rbrace$ and $L = \lbrace l_0,l_1,\dots,l_N \rbrace$. Note that the lengths of the sets $K$ and $L$ are the same. We will say that $L \ge K$ if $l_j \ge k_j$ for every index $j$.

Proving the stabilization properties we will have to work with the similarly defined coordinates and structures of the different Hurwitz--Frobenius manifolds. Due to this we will denote by
\[
    v_{i,\alpha}^K, t_{i,\alpha}^K, \quad v_{i,\alpha}^L, t_{i,\alpha}^L
\]
the $v_\bullet$ and $t_\bullet$ variables of $\H_{0,K}$ and $\H_{0,L}$ respectively. And we will denote by $F^{K}$, $F^{L}$  the potential of these them.
Note that the set of coordinates of $\H_{0,L}$ is bigger than the set of coordinates of $\H_{0,K}$.

Finally 
\[
 {}^{(K)}c_{i,a;j,b}^{k,l}, \quad \text{and} \quad {}^{(L)}c_{i,a;j,b}^{k,l}
\]
will stand for the product $\circ$ structure constants in the $v_\bullet$ frame of $\H_{0,K}$ and $\H_{0,L}$ respectively. Note again that the indices run over the different sets for the two manifolds.

\begin{proposition}\label{prop: psi stabilization}
    Let $L \ge K$. Then
    \[
        \frac{\p v_{i,\alpha}^K}{\p t_{j,\beta}^K} \mid_{t_{r,\gamma} = k_r s_{r,k_r-\gamma}}
         =
         \frac{\p v_{i,\alpha}^L}{\p t_{j,\beta}^L} \mid_{t_{r,\gamma} = l_r s_{r, l_r-\gamma}}
    \]
\end{proposition}
\begin{proof}    
    First assume $i \ge 1$. 
    One should only prove the proposition for $t_{i,\alpha}$ with $\alpha \ge 1$. 
    It follows immediately from the definition that $t_{i,\alpha}$ is a function of $v_{i,0},\dots,v_{i,k_i}$ only. 
    
    The coordinate $v_{i,a}$ has the following residue expression
    \begin{align*}
        v_{i,a} & = - \res_{z = q_i} ((z-q_i)^a \lambda dz) = - \res_{z = q_i} ((z-q_i)^a z_i^{k_i} dz)
        \\
        &= - \frac{1}{a+1} \res_{z = q_i} (z_i^{k_i} d (z-q_i)^{a+1}) 
        = \frac{1}{a+1} \res_{z = q_i} ((z-q_i)^{a+1} d z_i^{k_i} )
        \\
        & = \frac{k_i}{a+1} \res_{z_i = 0} (z_i^{k_i-1} (z-q_i)^{a+1} d z_i )
    \end{align*}
    In a neighborhood of $z = q_i$ by Eq.~\eqref{eq: flat coordinates i} holds (see also Lecture~5 and Eq.(5.51) of \cite{D2})
    \begin{equation}\label{eq: local variable expansion}
        z -q_i = \frac{1}{k_i} \sum_{j=1}^{k_i} t_{i,j} z_i^{-(k_i-j)} + o(z_i^{-k_i}).
    \end{equation}
    Then 
    \begin{align*}
        \frac{v_{i,a}}{t_{i,b}} &= - \res_{z_i = 0} (z_i^{k_i-1} z_i^{-(k_i-b)} (z-q_i)^{a} d z_i ) = - \res_{z_i = 0} (z_i^{b-1} (z-q_i)^{a} d z_i )
        \\
        &= - \res_{z_i = 0} \Big( z_i^{b-1} (\sum_{j=1}^{b} \frac{t_{i,k_i-j}}{k_i}  z_i^{-j} )^{a} d z_i \Big)
    \end{align*}
    where on the last step we have taken the summation only up to $j=b$ by the simple residue calculation.
    
    One notes now immediately that the expression we have obtained does not depend on $k_i$ after the substitution $t_{i,k_i-j} = k_i s_{i,j}$.

    The same proof can be adapted for $v_{0,\bullet}$ variables. However it also follows from Lemma~4.1 in \cite{BDbN}. 
\end{proof}

\begin{proposition}\label{prop: v-product stabilization}
    Let $L \ge K$, $i \neq j$ and $a < k_i$, $b < k_j$. Then in the basis $ \frac{\p}{\p v}$
    \[
        {}^{(K)}c_{i,a;j,b}^{k,l} = {}^{(L)}c_{i,a;j,b}^{k,l} \quad \forall k,l
    \]
    assumed at the equality of functions in $v_{j,\alpha}$.
\end{proposition}
\begin{proof}
    The statement is equivalent to the claim that
    \[
        \frac{1}{(z - q_i)^{a+1}} \circ \frac{1}{(z - q_j)^{b+1}}
    \]
    has the same expression via $\p \lambda / \p v_{k,l}$ in both Hurwitz--Frobenius manifolds. Because $i \neq j$ the rational function $\frac{1}{(z - q_i)^{a+1}} \cdot \frac{1}{(z - q_j)^{b+1}}$ belongs to both Hurwitz--Frobenius manifolds assumed and can therefore be expressed via $\p \lambda / \p v_{k,l}$.    
    
    The statement follows now by Proposition~\ref{prop: low ramification products} because in the $v$--coordinates the Hurwitz space $M_K$ is embedded into $M_L$ just by setting some of the coordinates to zero.
\end{proof}

\begin{theorem}\label{theorem: stabilization}
    Let $L \ge K$. For any $0 \le i,j \le o$ let $\alpha \le k_i$, $\beta \le k_j$ if $i \neq j$ and $\alpha+\beta \le k_i-2$ if $i=j$. 
    
    Then
    \[
        \frac{\p^2 F^{K}}{\p t_{i,\alpha} \p t_{j,\beta}} \mid_{t_{r,\gamma} = k_r s_{r,k_r-\gamma}}
        =
        \frac{\p^2 F^{L}}{\p t_{i,\alpha} \p t_{j,\beta}} \mid_{t_{r,\gamma} = k_r s_{r,k_r-\gamma}}
    \]
    assumed at the equality of function in $s_{r,\bullet}$ and $q_1,\dots,q_N$.
\end{theorem}
\begin{proof}
    Denote $\Psi_{\alpha}^a := \frac{\p v_a}{\p t_\alpha}$ the transition matrix between the essential and flat coordinate frames. Then $\frac{\p^3 F}{\p t_\alpha \p t_\beta \p t_\gamma} = \Psi_{\alpha}^a \Psi_{\beta}^b c_{a,b}^r (\Psi^{-1})_\delta^r \eta_{\delta \gamma}$ assuming the summation over all repeating indices. Here the tensors $c_{a,b}^r$ stand for the multiplication table components in the frame $\p / \p v_{i,\bullet}$.
    
    It was proved in Proposition~\ref{prop: psi stabilization} that the components of the matrices $\Psi$ stabilize. It was also proved in Proposition~\ref{prop: v-product stabilization}, Example~\ref{example: v-products in same sector} that the tensors $c_{a,b}^r$ stabilize. 
\end{proof}

\begin{corollary}
    The system of PDEs Eq.~\eqref{eq: BDbN form} defined by the series of Hurwitz--Frobenius manifolds, is commuting.
\end{corollary}
\begin{proof}
    Commutativity of the PDEs follows from WDVV equation exactly as in \cite[Section 3.2]{B24} and \cite[Proposition 2.1]{BDbN}.
\end{proof}

\section{Whitham and multi-KP hierarchies}
Introduce the genus zero Whitham hierarchy in the Fay form following Theorem~1 of \cite{TT07}.

For $i =0,1,\dots$ consider the operators
\[
 D_i(z) := \sum_{\alpha = 1}^\infty \frac{z^{-\alpha}}{\alpha} \frac{\p}{\p t_{i,\alpha}}.
\]
In what follows we will also abbreviate $\p_{i,\alpha} = \frac{\p}{\p t_{i,\alpha}}$.

Then genus zero Whitham hierarchy is the following system of commuting PDEs on the function $f$ depending on $\lbrace t_{i,\alpha} \rbrace$ where $\alpha = 0,1,2,\dots$  if $i \ge 1$ and $\alpha = 1,2,\dots$ if $\alpha = 0$.

\begin{align}
    e^{D_0(z)D_0(w)f} &= 1 - \frac{\p_{0,1} (D_0(z) - D_0(w)) f}{z - w},
    \label{Fay-00}
    \\
    z e^{D_0(z)(\p_{i,0} + D_i(w))f} &= z - \p_{0,1} (D_0(z) - \p_{i,0} - D_i(w)) f,
    \label{Fay-0i}
    \\
    e^{(\p_{i,0} +D_i(z))(\p_{i,0} +D_i(w))f} &= - \frac{zw \p_{0,1} (D_i(z) - D_i(w)) f}{z - w},
    \label{Fay-ii}
    \\
    \eps_{ij} e^{(\p_{i,0} + D_i(w))(\p_{j,0} + D_j(w))f} &= - \p_{0,1} (\p_{i,0}+ D_i(z) - \p_{j,0} - D_j(w)) f,
    \label{Fay-ij}
\end{align}
where $\eps_{ij} = 1$ if $i \le j$ and $\eps_{ij} = -1$ if $i > j$. The role of the multiple $\eps_{ij}$ is very simple - it realizes the symmetry of the both sides with respect to $i$ and $j$ interchange.

These equalities should be understood as the equalities of the formal power series in $z^{-1}$ and $w^{-1}$. Comparing the coefficents of $z^{-\alpha}w^{-\beta}$ on the both sides one gets the expression of the second order derivatives $\p_{i,\gamma}\p_{j,\delta}f$ via $\p_{0,1}\p_{i,\mu}f$.
This is important to note that the derivative with respect to $t_{i,0}$ never appears in any $D_i(z)$.

\begin{theorem}\label{theorem: Fay}
    Consider the Hurwitz--Frobenius manifold $\H_{0,k_0,k_1,\dots,k_N}$ with the potential $F$ and denote $m := \min(k_0,k_1,\dots,k_N)$.
    
    The potential $F$ satisfies equations \eqref{Fay-00}, \eqref{Fay-0i}, \eqref{Fay-ii} and \eqref{Fay-ij} up to total order $-m$ of the formal variables $z$ and $w$ after the rescaling $t_{i,\alpha} \mapsto \alpha t_{i,\alpha}$ with $ i \ge 0$ and $\alpha \ge 1$.
\end{theorem}
\begin{proof}
    Is given in Section~\ref{section: proof}.
\end{proof}
\begin{corollary}
    The infinite system of PDEs \eqref{eq: BDbN form} defined by the stabilizing series of Hurwitz--Frobenius manifolds $\H_{0,K}$ coincides with the dispersionless genus zero Whitham hierarchy.
\end{corollary}
\begin{proof}
    The right hand side of Eq.~\eqref{eq: BDbN form} is obtained by the substitution $t_{r,\gamma} = -k_r \p_{0,1} \p_{r,k_r-\gamma} f$ to the second order derivatives of the potential, that are rational functions in the flat coordinates. So, the ``right'' flat coordinate expansion of $F^K$ is exactly what gives the proof.

    Taking $f = F^K$ we have $-k_r \p_{0,1} \p_{r,k_r-\gamma} F^K = t_{r,\gamma}$ and the right hand side of Eq.~\eqref{eq: BDbN form} is just the expansion of the second order derivatives of $F^K$. Now to check that $F^K$ satisfies the certain equation from equations \eqref{Fay-00}, \eqref{Fay-0i}, \eqref{Fay-ii} is the same as to check that the right hand side of Eq.~\eqref{eq: BDbN form} coincides with the respective Fay--form equation.
\end{proof}

\begin{remark}\label{remark: coordinates sign}
    It is exactly this theorem, where the right joice of the coordinates is important (recall Section~\ref{section: flat structure}). Changing the sign of $t_{0,\bullet}$ does not affect Eq.~\eqref{Fay-00} but results in the sign change of the other equations. We choose to fix the flat coordinates so that the Fay form equations coincide with those of Takasaki-Takebe.
\end{remark}

\subsection{h--deformation}\label{section: h-deformation}
Theorem~2 of \cite{TT07} makes us to propose the $\hbar$--deformation of the genus zero Whitham hierarchy. Namely, Takasaki and Takebe proved that Fay form equations of the genus zero Whitham hieararchy are obtained as the dispersionless limit of the following system of equations on $\tau = \tau(t)$ that they derive from the multiKP hierarchy.
\begin{align*}
    \exp \left( (e^{D_0(z)}-1)(e^{D_0(w)}-1) \log \tau \right) 
    &= 1 - \frac{\p_{0,1} (e^{D_0(z)} - e^{D_0(w)}) \log \tau}{z - w},
    \\
    z \exp \left( (e^{D_0(z)}-1)(e^{\p_{i,0} + D_i(w)}-1) \log \tau  \right) 
    &= z - \p_{0,1} (e^{D_0(z)} - e^{\p_{i,0} + D_i(w)}) \log \tau,
    \\
    \exp \left( (e^{\p_{i,0} + D_i(z)} - 1) (e^{\p_{i,0} + D_i(w)}-1) \log \tau \right)
    &= - \frac{zw \p_{0,1} (e^{\p_{i,0} + D_i(z)} - e^{\p_{i,0} + D_i(w)} \log \tau }{z - w},
    \\
    \eps_{ij} \exp \left( (e^{\p_{i,0} + D_i(z)}-1)(e^{\p_{j,0} + D_j(w)}-1) \log \tau \right)
    &= - \p_{0,1} (e^{\p_{i,0} + D_i(z)} - e^{\p_{j,0} + D_j(w)}) \log \tau.
\end{align*}
Here the operator $e^{\p_{i,0}}$ should be understood as the variable shift 
$t_{i,0} \mapsto t_{i,0}+1$.

One notes immediately that these equations are obtained from the Fay form equations of the genus zero Whitham hierarchy by the substitution
\[
    D_0(z) \mapsto e^{D_0(z)} - 1, \quad \p_{i,0} + D_i(z) \mapsto e^{\p_{i,0} + D_i(z)} - 1.
\]

This substitution can be used to derive the ``dispersionfull'' hierarchy from the system of PDE's that we get from the family of Hurwitz--Frobenius manifolds $\H_{0,K}$.

\subsection{Proof of Theorem~\ref{theorem: Fay}}\label{section: proof}
Due to Theorem~\ref{theorem: stabilization} we may consider only the Hurwitz--Frobenius manifold $\H_{0,K}$ with $K = \lbrace k,\dots, k \rbrace$ for all $k \ge 1$ and fixed length $|K|$.
All the computations of this section will be done in this Dubrovin--Frobenius manifold.

Consider the quotient-ring 
\[ 
\A_k := \QQ[\bt]\otimes\QQ[[u^{-1},v^{-1}]]\, \big/\, (u^{-k},u^{-(k-1)}v^{-1},  \dots, u^{-1}v^{-(k-1)},v^{-k}).
\]
Namely, this is the finite rank $\QQ[\bt]$--module generated by polynomials in $u^{-1}$ and $v^{-1}$ with the total degree not exceeding $k$.

Consider the formal power series $p_\bullet(u)$.
\begin{align}
    &p_0(u) := u - \sum_{\alpha=1}^{k-1} u^{-\alpha} \p_{0,1}\p_{0,\alpha}F,
    \\
    &p_i (u) := -\p_{0,1}\p_{i,0}F  - \sum_{\alpha=1}^{k-1} u^{-\alpha} \p_{0,1}\p_{i,\alpha}F , \quad i \ge 1.
\end{align}

To prove the theorem we should show Eq.~\eqref{Fay-00}, \eqref{Fay-0i}, \eqref{Fay-ii} and \eqref{Fay-ij}.
Eq.~\eqref{Fay-00} was proved to hold in \cite{BDbN}.
Eq.~\eqref{Fay-0i} is proved in Propositions~\ref{prop: Eq. 39}, Eq.~\eqref{Fay-ii} is proved in Proposition~\ref{prop: Eq. 40} and Eq.~\eqref{Fay-ij} is proved in Proposition~\ref{prop: Eq. 41}.

\begin{proposition}\label{prop: Eq. 39}
    In $\A_k$ holds
    \begin{align}
        & p_0(u) - t_{i,k} = u \exp\left( \sum_{\alpha=1}^{k-1} u^{-\alpha} \p_{0,\alpha}\p_{i,0}F \right),
        \label{eq: 39 part}
        \\
        & p_0(u) - p_i(v) = (p_0(u) - t_{i,k}) \exp\left( \sum_{\alpha,\beta=1}^{k-1} u^{-\alpha} v^{-\beta} \p_{0,\alpha}\p_{i,\beta} F \right).
        \label{eq: 39 full}
    \end{align}
\end{proposition}
\begin{proof}
We first prove Eq.~\eqref{eq: 39 part}.

\begin{lemma}
    Let $r < k$ and $r > 1$. Then
    \begin{align*}
        & r \p_{0,r}\p_{i,0} F
        = r \frac{t_{0,k-(r-1)}}{k}  + t_{i,k} (r-1) \p_{0,r-1}\p_{i,0} F - \sum_{b=1}^{r-2} \frac{t_{0,k-b}}{k} (r-1-b) \p_{0,r-1-b} \p_{i,0} F.
    \end{align*}
\end{lemma}
\begin{proof}
    Note that $\frac{\p \lambda}{\p t_{i,0}} = (z - t_{i,k})^{-1}$.
    We have by using Eq.~\eqref{eq: lambda 1}
    \begin{align*}
        & \frac{\p \lambda}{\p t_{0,r}} \frac{\p \lambda}{\p t_{i,0}} = 
        z \frac{\p \lambda}{\p t_{0,r-1}} \frac{\p \lambda}{\p t_{i,q}} 
        - \sum_{b=1}^{r-2} \frac{t_{0,k-b}}{k} \frac{\p \lambda}{\p t_{0,r-1-b}} \frac{\p \lambda}{\p t_{i,0}} 
        \\
        & = \frac{\p \lambda}{\p t_{0,r-1}}  
        + t_{i,k} \frac{\p \lambda}{\p t_{0,r-1}} \frac{\p \lambda}{\p t_{i,0}} 
        - \sum_{b=1}^{r-2} \frac{t_{0,k-b}}{k} \frac{\p \lambda}{\p t_{0,r-1-b}} \frac{\p \lambda}{\p t_{i,0}}.
    \end{align*}
    The statement follows now from Lemma~\ref{lemma: sd}.    
\end{proof}

Consider the polynomials
\begin{align*}
    & A_0 := \sum_{r=1}^k u^{-r} \cdot r \p_{0,r}\p_{i,0} F,
    \\
    & A_1 := u^{-1} \left( -t_{i,k} + \sum_{q=2}^k u^{-(q-1)} \frac{q}{k} t_{0,k+1-q}\right),
    \quad B_1 := u^{-1} \left(-t_{i,k} + \sum_{q\ge 2} u^{-(q-1)} \frac{1}{k} t_{0,k+1-q} \right).
\end{align*}

It follows from lemma above that
\[
 A_0  = A_1 - A_0 \cdot B_1  \quad \Leftrightarrow \quad A_0 = \frac{A_1}{1 + B_1}.
\]

On the other side note that we have $p_0(u) - t_{i,k} = u (1 + B_1)$. Applying the operator $E = - u \frac{\p}{\p u} - v \frac{\p}{\p v}$ to both sides, Eq.~\eqref{eq: 39 part} is equivalent to
\[
    A_1 = A_0 (1 + B_1)
\]
that we have just observed above.

Now consider Eq.~\eqref{eq: 39 full}.

\begin{lemma}
    Let $p+q+1 \le k$ and $q > 0$. Then
    \begin{align*}
        & (p+q+1)\p_{0,p+1}\p_{i,q} F =
        \\
        &\quad = t_{1,k} (p+q) \p_{0,p}\p_{i,q} F - \sum_{b=1}^{p-1} \frac{t_{0,k-b}}{k} (p+q-b) \p_{0,p-b} \p_{i,q} F + \sum_{b=1}^q \frac{t_{i,k-b}}{k} (p+q-b) \p_{0,p} \p_{i,q-b} F 
    \end{align*}
\end{lemma}
\begin{proof}
    We have by using Eq.~\eqref{eq: lambda 1}
    \begin{align*}
        & \frac{\p \lambda}{\p t_{0,p+1}} \frac{\p \lambda}{\p t_{i,q}} = 
        z \frac{\p \lambda}{\p t_{0,p}} \frac{\p \lambda}{\p t_{i,q}} 
        - \sum_{b=1}^{p-1} \frac{t_{0,k-b}}{k} \frac{\p \lambda}{\p t_{0,p-b}} \frac{\p \lambda}{\p t_{i,q}} 
        \\
        & = (z - t_{i,k}) \frac{\p \lambda}{\p t_{0,p}} \frac{\p \lambda}{\p t_{i,q}} 
        + t_{i,k} \frac{\p \lambda}{\p t_{0,p}} \frac{\p \lambda}{\p t_{i,q}} 
        - \sum_{b=1}^{p-1} \frac{t_{0,k-b}}{k} \frac{\p \lambda}{\p t_{0,p-b}} \frac{\p \lambda}{\p t_{i,q}}.
    \end{align*}
    Now by using Eq.~\eqref{eq: lambda 2} this is equal to
    \begin{align*}
        & \sum_{b=1}^{q} \frac{t_{i,k-b}}{k} \frac{\p \lambda}{\p t_{0,p}} \frac{\p \lambda}{\p t_{i,q-b}} 
        + t_{1,k} \frac{\p \lambda}{\p t_{0,p}} \frac{\p \lambda}{\p t_{i,q}} 
        - \sum_{b=1}^{p-1} \frac{t_{0,k-b}}{k} \frac{\p \lambda}{\p t_{0,p-b}} \frac{\p \lambda}{\p t_{i,q}}.
    \end{align*}
    The statement follows now from Lemma~\ref{lemma: sd}.
    
\end{proof}

Consider the polynomials
\begin{align*}
    & A_0 := \sum_{r=1}^k\sum_{q=0}^k u^{-r} v^{-q} (r+q) \p_{0,r}\p_{1,q} F,
    \\
    & A_1 := u^{-1} \left( -t_{1,k} + \sum_{q=2}^k u^{-(q-1)} \frac{q}{k} t_{0,k+1-q} - \sum_{q=2}^k v^{-(q-1)} \frac{q}{k} t_{1,k+1-q} \right),
    \\
    & B_1 := u^{-1} \left(-t_{1,k} + \sum_{q\ge 2} u^{-(q-1)} \frac{1}{k} t_{0,k+1-q} -  \sum_{q \ge 2} v^{-(q-1)} \frac{1}{k} t_{1,k+1-q} \right).
\end{align*}

It follows from lemma above that
\[
 A_0  = A_1 - A_0 \cdot B_1  \quad \Leftrightarrow \quad A_0 = \frac{A_1}{1 + B_1}.
\]

On the other side note that we have $p_0(u) - p_i(v) = u(1 + B_1)$.

Eq.~\eqref{eq: 39 full} after Eq.~\eqref{eq: 39 part} reads
\[
    p_0(u) - p_i(v) = u \exp\left( \sum_{\alpha=1}^{k-1}\sum_{\beta=0}^{k-1} u^{-\alpha} v^{-\beta} \p_{0,\alpha}\p_{i,\beta} F \right).
\]
Applying the operator $E = -u \frac{\p}{\p u} - v \frac{\p}{\p v}$ to both sides of this equality we see that it is equivalent to
\[
    A_1 = A_0 (1 + B_1)
\]
that we have just observed above.

\end{proof}

\begin{proposition}\label{prop: Eq. 41}
    For any $i,j \ge 1$ and $i \neq j$ in $\A_k$ holds
    \[
        p_i(u) - p_j(v) = (t_{i,k} - t_{j,k}) \exp\left( \sum_{\alpha,\beta=1}^{k-1} u^{-\alpha} v^{-\beta} \p_{i,\alpha}\p_{j,\beta} F + \sum_{\alpha=1}^{k-1} (u^{-\alpha} \p_{i,\alpha}\p_{j,0} F + v^{-\alpha} \p_{j,\alpha}\p_{i,0} F )\right).
    \]
\end{proposition}
\begin{proof}
Note that the desired equality holds in the zeros power in $u^{-1}$ and $v^{-1}$ due to Eq.~\ref{prop: simple 0 insertions}.

We have by using Eq.~\eqref{eq: lambda 2}
\begin{align}
    & \sum_{r=1}^b \frac{t_{j,k-r}}{k} \frac{\p \lambda}{\p t_{i,a}} \frac{\p \lambda}{\p t_{j,b-r}} 
    = (z - t_{j,k}) \frac{\p \lambda}{\p t_{i,a}}\frac{\p \lambda}{\p t_{j,b+1}} 
    \\
    &\quad = (z - t_{i,k}) \frac{\p \lambda}{\p t_{i,a}}\frac{\p \lambda}{\p t_{j,b+1}} 
    + (t_{i,k} - t_{j,k}) \frac{\p \lambda}{\p t_{i,a}}\frac{\p \lambda}{\p t_{j,b+1}} 
    \\
    &\qquad = \sum_{r=1}^a \frac{t_{i,k-r}}{k} \frac{\p \lambda}{\p t_{i,a-r}}\frac{\p \lambda}{\p t_{j,b+1}} 
    + (t_{i,k} - t_{j,k}) \frac{\p \lambda}{\p t_{i,a}}\frac{\p \lambda}{\p t_{j,b+1}},
\end{align}
where we used Eq.~\eqref{eq: lambda 2} again in the last step.

This follows from Lemma~\ref{lemma: sd} that equality above is equivalent to the following power series equality
\begin{align*}
     & \left(\sum_r (a+b-r) u^{-(a-r)}v^{-b} \p_{\alpha,a}\p_{\beta,b-r}F\right) \cdot \sum_r \frac{t_{\alpha,k-r}}{k} v^{-r} 
     \\
     & \quad - \left(\sum_r (a+b-r) u^{-a}v^{-(b-r)} \p_{\alpha,a-r}\p_{\beta,b}F\right) \cdot \sum_r \frac{t_{\beta,k-r}}{k} u^{-r} 
     \\
     & \quad\qquad = (t_{\alpha,k} - t_{\beta,k}) \left(\sum_r (a+b) u^{-a}v^{-b} \p_{\alpha,a}\p_{\beta,b}F\right).
\end{align*}
One notes immediately that this is equivalent to the desired equation after applying the operator $E = -u \frac{\p}{\p u} - v \frac{\p}{\p v}$.

\end{proof}

\begin{proposition}\label{prop: Eq. 40}
    For any $i \ge 1$ in $\A_k$ holds
    \begin{align}
        & p_i(u) - t_{i,k} = u^{-1} \frac{t_{i,k-1}}{k} \exp\left( \sum_{\alpha=0}^{k-1} u^{-\alpha} \p_{i,\alpha}\p_{i,0}F \right),
        \label{eq: 40-part}
        \\
        & (p_i(u) - p_i(v)) = \frac{t_{i,k-1}}{k} (u^{-1}-v^{-1}) \exp\left( \sum_{\alpha,\beta=1}^{k-1} u^{-\alpha} v^{-\beta} \p_{i,\alpha}\p_{i,\beta} F + \sum_{\alpha=1}^{k-1} (u^{-\alpha} + v^{-\alpha} ) \p_{i,\alpha}\p_{i,0} F  \right).
        \label{eq: 40-full}
    \end{align}

\end{proposition}
\begin{proof}
is similar to the proof of Proposition~\ref{prop: Eq. 41}.

First of all note that Eq.~\eqref{eq: 40-part} holds in the first order in $u^{-1}$ because $\p_{0,1}\p_{i,1}F = - t_{i,k-1}/k$. Proof of this equality in the higher order in $u^{-1}$ is complitely similar to the proof of Eq.~\eqref{eq: 39 part}

To show Eq.~\eqref{eq: 40-full} we have by using Eq.~\eqref{eq: lambda 2} twice
\begin{align}
    & \sum_{r=1}^b \frac{t_{i,k-r}}{k} \frac{\p \lambda}{\p t_{i,a}} \frac{\p \lambda}{\p t_{i,b-r}} 
    = (z - t_{i,k}) \frac{\p \lambda}{\p t_{i,a}}\frac{\p \lambda}{\p t_{i,b+1}} 
    = \sum_{r=1}^a \frac{t_{i,k-r}}{k} \frac{\p \lambda}{\p t_{i,a-r}}\frac{\p \lambda}{\p t_{i,b+1}}.
\end{align}
\end{proof}

\end{document}